\documentclass[
	paper=a4,%
	11pt,american,pagesize,%
	DIV=12,headinclude=true,headlines=0,footinclude=true,footlines=-5,twoside=semi%
]{scrartcl}
\def\version{arxiv}
	\usepackage{hyperref}

\def\draftmode{false}

\usepackage[utf8]{inputenc}
\RequirePackage{etex}
\RequirePackage{fixltx2e}

\makeatletter
\usepackage{lmodern}
\usepackage[T1]{fontenc}

\usepackage{babel}
\input{ushyphex.tex} 

\newcommand{\ifarxiv}[2]{\ifthenelse{\equal{\version}{arxiv}}{#1}{#2}}
\newcommand{\ifconf}[2]{\ifthenelse{\equal{\version}{conference}}{#1}{#2}}
\newcommand{\ifdraft}[2]{\ifthenelse{\equal{\draftmode}{true}}{#1}{#2}}

\usepackage{array,multicol}
\usepackage{amsmath,amsfonts,amssymb,mathtools,colonequals}
\usepackage{mleftright}\mleftright 
\usepackage{relsize,ifthen,xspace,enumitem,booktabs,adjustbox,needspace}
\usepackage{graphicx}

\usepackage{tikz}
\usetikzlibrary{arrows.meta,calc,positioning}

\AtBeginDocument{%
	\let\mytitle\@title%
}

	\usepackage{colonequals}
	\usepackage[headsepline]{scrlayer-scrpage}
	\clearscrheadfoot
	\AtBeginDocument{%
		\manualmark%
		\markleft{\mytitle}%
		\automark*[section]{}%
	}
	\ohead{\pagemark}
	\ihead{\headmark}
	\addtokomafont{caption}{\sffamily\small}
	\addtokomafont{captionlabel}{\sffamily\textbf}
	\setcapmargin{2em}

\hypersetup{
	final,
	unicode=true, 
	bookmarks=true,
	bookmarksnumbered=true,
	bookmarksdepth=2,
	bookmarksopen=true,
	breaklinks=true,
	hidelinks,
	pdfhighlight={/P}
}
\usepackage{microtype}

\usepackage{wref}

\usepackage{clrscode3e}

	\usepackage{newfloat}
	\DeclareFloatingEnvironment[%
			name=Algorithm,%
			placement=thb,%
		]{algorithm}

\usepackage[amsmath,hyperref,thmmarks]{ntheorem}

\theorembodyfont{\slshape}
\theoremseparator{:}
\newtheoremstyle{proofstyle}%
  {\item[\theorem@headerfont\hskip\labelsep ##1\theorem@separator]}%
  {\item[\theorem@headerfont\hskip\labelsep ##1 of ##3\theorem@separator]}

\theorempreskip{\topsep} 

\newtheorem{theorem}{Theorem}[section]

\theoremstyle{plain}
\theorempreskip{\topsep}

\newtheorem{lemma}[theorem]{Lemma}

\theoremstyle{plain}
\theorembodyfont{\upshape}

\theoremsymbol{\raisebox{-.25ex}{$\Box$}}
\qedsymbol{\raisebox{-.25ex}{$\Box$}}

\theoremstyle{proofstyle}
\newtheorem{proof}{Proof}

\newenvironment{thmenumerate}[1]{%
	\begin{enumerate}[
		label={\makebox[\widthof{(a)}][c]{\textup{(\alph*)}}},
		ref={\ref{#1}\kern.1em--\kern.1em(\alph*)},
		itemsep=0pt,
	]%
}{%
	\end{enumerate}%
}

\newdimen\makeboxdimen
\newcommand\makeboxlike[3][l]{%
\setbox0=\hbox{#2}%
\global\makeboxdimen=\wd0%
\setbox1=\hbox{\makebox[\makeboxdimen][#1]{%
\makebox[0pt][#1]{#3}%
}}%
\ht1=\ht0%
\dp1=\dp0%
\box1%
}

\newcommand\plaincenter[1]{%
	\ignorespaces%
		\mbox{}\hfill#1\hfill\mbox{}
	\unskip%
}

\newcommand*\ie{\mbox{i.\hspace{.2ex}e.}}
\newcommand*\eg{\mbox{e.\hspace{.2ex}g.}}

\newcommand*\wrt{\mbox{w.\hspace{.2ex}r.\hspace{.2ex}t.}\xspace}

\newcommand*\withoutlossofgenerality{\mbox{w.\hspace{.23ex}l.\hspace{.2ex}o.\hspace{.17ex}g.}\xspace}

\newcommand{\ESymbol}{\mathbb{E}}

\newcommand\R{\mathbb R}
\newcommand\N{\mathbb N}
\newcommand\Z{\mathbb Z}

\newcommand\Oh{O}
\def\.{\mskip1mu}

\newcommand{\ProbSymbol}{\ensuremath{\mathbb{P}}}
\providecommand{\given}{}
\DeclarePairedDelimiterXPP\Prob[1]{\ProbSymbol}[]{}{%
	\renewcommand\given{\nonscript\:\delimsize\vert\nonscript\:\mathopen{}}%
	#1%
}
\DeclarePairedDelimiterXPP\E[1]{\ESymbol}[]{}{%
	\renewcommand\given{\nonscript\:\delimsize\vert\nonscript\:\mathopen{}}%
	#1%
}
\DeclarePairedDelimiterXPP\Eover[2]{\ESymbol_{#1}}[]{}{%
	\renewcommand\given{\nonscript\:\delimsize\vert\nonscript\:\mathopen{}}%
	#2%
}
\providecommand{\Prob}{} 
\providecommand{\E}{} 
\providecommand{\Eover}{} 

\newcommand\ui[2]{#1^{\smash{(}#2\smash{)}}}

\usepackage{fixmath}

\newcommand\eqdist{	
	\mathchoice{
		\mathrel{\overset{\raisebox{0ex}{$\scriptstyle \mathcal D$}}=}%
	}{
		\mathrel{\like{=}{%
			\overset{\raisebox{-1ex}{$\scriptscriptstyle \mathcal D$}}=%
		}}%
	}{
		\mathrel{\overset{\mathcal D}=}%
	}{
		\mathrel{\overset{\mathcal D}=}%
	}%
}

\newcommand\ppe{\phantom{=}}

\newcommand\like[3][c]{%
	\mathchoice{
		\makeboxlike[#1]{%
			\ensuremath{\displaystyle\relax#2}%
		}{%
			\ensuremath{\displaystyle\relax#3}%
		}%
	}{
		\makeboxlike[#1]{%
			\ensuremath{\textstyle\relax#2}%
		}{%
			\ensuremath{\textstyle\relax#3}%
		}%
	}{
		\makeboxlike[#1]{%
			\ensuremath{\scriptstyle\relax#2}%
		}{%
			\ensuremath{\scriptstyle\relax#3}%
		}%
	}{
		\makeboxlike[#1]{%
			\ensuremath{\scriptscriptstyle\relax#2}%
		}{%
			\ensuremath{\scriptscriptstyle\relax#3}%
		}%
	}
}

\newcommand\uniform{\mathcal U}

\newcommand\binomial{\mathrm{Bin}}

\newcommand\betadist{\mathrm{Beta}}

\newcommand\betaBinomial{\mathrm{BetaBin}}

\DeclarePairedDelimiterXPP\distFromWeights[1]{\mathcal D}(){}{#1}
\providecommand\distFromWeights{} 

\newcommand\BetaFun{\mathrm B}

\renewcommand\given{\mathbin{\mid}}
\newcommand\harm[1]{\ensuremath{H_{#1}}}

\newcommand\ce{\colonequals}

\newcommand{\surroundedmath}[3]{
	\mathchoice{
		#1{#2{#3}#2}%
	}{
		#1{#3}%
	}{
		#1{#3}%
	}{
		#1{#3}%
	}%
}

\newcommand\rel[1]{\surroundedmath{\mathrel}{\:}{#1}}
\newcommand\wrel[1]{\surroundedmath{\mathrel}{\;}{#1}}
\newcommand\wwrel[1]{\surroundedmath{\mathrel}{\;\;}{#1}}
\newcommand\bin[1]{\surroundedmath{\mathbin}{\:}{#1}}
\newcommand\wbin[1]{\surroundedmath{\mathbin}{\;}{#1}}
\newcommand\wwbin[1]{\surroundedmath{\mathbin}{\;\;}{#1}}

\newcommand{\eqwithref}[2][c]{%
	\relwithref[#1]{#2}{=}%
}
\newcommand{\relwithref}[3][c]{%
	\mathrel{\underset{\mathclap{\makebox[\widthof{$=$}][#1]{\scriptsize\wref{#2}}}}{#3}}%
}

\makeatletter
\let\oldalign\align
\let\endoldalign\endalign

\newcommand*\numberthis[1][]{\stepcounter{equation}\tag{\theequation}}

\renewenvironment{align}{%
	\begingroup%
	\let\oldhalign\halign
	\def\halign{%
		\let\oldbreak\\%
		\def\nonnumberbreak{\nonumber\oldbreak*}%
		\def\\{%
			\@ifstar{\nonnumberbreak}{\oldbreak}%
		}%
		\oldhalign%
	}
	\oldalign%
}{%
	\endoldalign%
	\endgroup%
}

\makeatother

\def\mydots{\xleaders\hbox to.5em{\hfill.\hfill}\hfill}

\newlength\tmpLenNotations
\newenvironment{notations}[1][10em]{%
	\small
	\newcommand\notationentry[1]{%
		\settowidth\tmpLenNotations{##1}%
		\ifthenelse{\lengthtest{\tmpLenNotations > \labelwidth}}{%
			\parbox[b]{\labelwidth}{%
				\makebox[0pt][l]{##1}\\%
			}%
		}{%
			\mbox{##1}%
		}%
		\mydots\relax%
	}%
	\begin{list}{}{%
		\setlength\labelsep{0em}%
		\setlength\labelwidth{#1}%
		\setlength\leftmargin{\labelwidth+\labelsep+1em}%
	}
	\ifdraft{%
		\newcommand\notation[1]{%
			\item[{##1}] \marginpar{%
				\adjustbox{scale={.55}{1},outer}{%
					\color{brown!80}%
					\scriptsize$\triangleright$\;\tiny\texttt{\detokenize{##1}}%
				}%
			}%
		}%
	}{%
		\newcommand\notation[1]{\item[{##1}]}%
	}%
	\raggedright
}{%
	\end{list}
}

\allowdisplaybreaks[3]

	\let\oldparagraph\paragraph
	\renewcommand\paragraph[1]{%
		\oldparagraph{#1.}
	}

\let\oldthebibliography\thebibliography
\renewcommand\thebibliography[1]{%
	\oldthebibliography{#1}%
	\pdfbookmark[1]{References}{}%
}

\def\myacknowledgements{}
\ifarxiv{
	\newcommand\acknowledgements[1]{\def\myacknowledgements{\paragraph{Acknowledgements}#1}}
}{}
\let\epsilon\varepsilon

\makeatother

\usepackage{hyperref}
\hypersetup{hidelinks}

\newcommand\myfulltitle{%
	Average Cost of QuickXsort with Pivot~Sampling
}

\title{\myfulltitle}
	\author{%
		Sebastian Wild%
		\thanks{%
			David R.\ Cheriton School of Computer Science, 
			University of Waterloo,
			Email: \texttt{wild\,@\,uwaterloo.ca}\protect\\
			This work was supported by the 
			Natural Sciences and Engineering Research Council of Canada 
			and the Canada Research Chairs Programme.
		}
	}

\acknowledgements{%
	I would like to thank three anonymous referees for many helpful comments, 
	references and suggestions 
	that helped improve the presentation of this paper.%
}

\makeatletter
\hypersetup{
	pdftitle={\@title},
	pdfauthor={\@author}
}

\begin{document}

\maketitle

\begin{abstract}
\noindent
\textbf{Abstract:}
QuickXsort is a strategy to combine Quicksort with another sorting method~X, so
that the result has essentially the same comparison cost as X in isolation, but
sorts in place even when X requires a linear-size buffer.
We solve the recurrence for QuickXsort precisely up to the linear term
including the optimization to choose pivots from a sample of $k$ elements.
This allows to immediately obtain overall average costs using only the average costs of 
sorting method X (as~if run in isolation).
We thereby extend and greatly simplify the analysis of
QuickHeapsort and QuickMergesort with practically efficient pivot selection, 
and give the first \emph{tight} upper bounds including the linear term for such methods.
\end{abstract}

\section{Introduction}

In QuickXsort~\cite{EdelkampWeiss2014}, we use the recursive scheme of ordinary Quicksort, but
instead of doing two recursive calls after partitioning, we first sort one of the segments 
by some other sorting method X.
Only the second segment is recursively sorted by QuickXsort.
The key insight is that X can use the second segment as a temporary buffer for elements.
By that, QuickXsort is sorting in-place (using $O(1)$ words of extra space)
even when X itself is not.

Not every method makes a suitable `X'; it must use the buffer in a swap-like fashion:
After X has sorted its segment, the 
elements originally stored in our buffer must still be intact,
\ie, they must still be stored in the buffer, albeit in a different order.
Two possible examples 
that use extra space in such a way 
are
Mergesort (see \wref{sec:quickmergesort} for details) and a comparison-efficient 
Heapsort variant~\cite{CantoneCincotti2002} with an output buffer.
With QuickXsort we can make those methods sort in-place
while retaining their comparison efficiency.
(We lose stability, though.)

While other comparison-efficient in-place sorting methods are 
known (\eg~\cite{Reinhardt1992,Katajainen1998,GeffertGajdos2011}),
the ones based on QuickXsort and elementary methods X are particularly easy to implement%
\footnote{%
	See for example the code for QuickMergesort that was presented for discussion 
	on code review stack exchange, \url{https://codereview.stackexchange.com/q/149443},
	and the succinct C++ code in~\cite{EdelkampWeiss2018}.
}
since one can adapt existing implementations for X.
In such an implementation, the tried and tested optimization to choose the pivot 
as the median of a small sample suggests itself to improve QuickXsort.
In previous works~\cite{CantoneCincotti2002,EdelkampWeiss2014,DiekertWeiss2016,EdelkampWeiss2018}, 
the influence of QuickXsort on the performance of X was either studied
by ad-hoc techniques that do not easily apply with general pivot sampling or
it was studied for the case of very good pivots:
exact medians or medians of a sample of $\sqrt n$ elements.
Both are typically detrimental to the average performance since they add significant
overhead, whereas most of the benefit of sampling is realized already for samples of very 
small constant sizes like $3$, $5$ or $9$.
Indeed, in a very recent manuscript~\cite{EdelkampWeiss2018}, Edelkamp and Weiß describe
an optimized median-of-3 QuickMergesort implementation in C++ that outperformed the
library Quicksort in \texttt{std::sort}.

The contribution of this paper is
\textbf{\boldmath a general transfer theorem (\wref{thm:cn}) that expresses the costs of QuickXsort
with median-of-$k$ sampling (for any odd constant $k$) directly in terms of the costs
of X,} (\ie, the costs that X needs to sort $n$ elements in isolation).
We thereby obtain the first analyses of QuickMergesort and QuickHeapsort 
with best possible constant-coefficient bounds on the linear term under realistic sampling schemes.

Since Mergesort only needs a buffer for one of the two runs, QuickMergesort
should not simply give Mergesort the smaller of the two segments to sort, 
but rather the \emph{largest one
for which the other segments still offers sufficient buffer space.}
(This will be the larger segment of the two if the smaller one contains at least a third of the elements;
see \wref{sec:quickmergesort} for details.)
Our transfer theorem covers this refined version of QuickMergesort, as well,
which had not been analyzed before.%
\footnote{%
	Edelkamp and Weiß do consider this version of QuickMergesort~\cite{EdelkampWeiss2014}, 
	but only analyze it for median-of-$\sqrt n$ pivots.
	In this case, the behavior coincides with the simpler strategy to 
	always sort the smaller segment by Mergesort
	since the segments are of almost equal size with high probability.%
}

The rest of the paper is structured as follows:
In \wref{sec:previous-work}, we summarize previous work on QuickXsort with 
a focus on contributions to its analysis.
\wref{sec:prelims} collects mathematical facts and notations used later.
In \wref{sec:quickxsort} we define QuickXsort and formulate a recurrence for its cost.
Its solution is stated in \wref{sec:main-result}.
\wref{sec:quickmergesort} presents the QuickMergesort as our stereotypical instantiation
of QuickXsort.
The proof of the transfer spreads over
\wref[Sections]{sec:solution-leading-term} and~\ref{sec:solution-linear-term}.
In \wref{sec:discussion}, we apply our result to QuickHeapsort and QuickMergesort
and discuss some algorithmic implications.

\section{Previous Work}
\label{sec:previous-work}
The idea to combine Quicksort and a secondary sorting method was suggested by
Contone and Cincotti~\cite{CantoneCincotti2000,CantoneCincotti2002}.
They study Heapsort with an output buffer (external Heapsort),%
\footnote{%
	Not having to store the heap in a
	consecutive prefix of the array allows to save
	comparisons over classic in-place Heapsort:
	After a delete-max operation, we can fill the gap at the root of the heap
	by promoting the largest child and recursively moving the gap down the heap.
	(We then fill the gap with a $-\infty$ sentinel value).
	That way, each delete-max 
	needs exactly $\lfloor \lg n\rfloor $ comparisons.%
}
and combine it with Quicksort to QuickHeapsort.
They analyze the average costs for external Heapsort in isolation and use 
a differencing trick for dealing with the QuickXsort recurrence;
however, this technique is hard to generalize to median-of-$k$ pivots.

Diekert and Weiß~\cite{DiekertWeiss2016} suggest optimizations for QuickHeapsort
(some of which need extra space again), and
they give better upper bounds for QuickHeapsort with random pivots and median-of-3.
Their results are still not tight since they upper bound the total cost of all Heapsort
calls together (using ad hoc arguments on the form of the costs for one Heapsort round),
without taking the actual subproblem sizes into account that Heapsort is used on.
In particular, their bound on the overall contribution of the Heapsort calls does not depend
on the sampling strategy.

Edelkamp and Weiß~\cite{EdelkampWeiss2014} explicitly describe QuickXsort as a general
design pattern and, among others, consider using Mergesort as `X'.
They use the median of $\sqrt n$ elements
in each round throughout to guarantee good splits with high probability.
They show by induction that when X uses at most
$n \lg n + cn + o(n)$ comparisons on average for some constant $c$,
the number of comparisons in QuickXsort is also bounded by $n \lg n + cn + o(n)$.
By combining QuickMergesort with Ford and Johnson's MergeInsertion~\cite{FordJohnson1959}
for subproblems of logarithmic size, Edelkamp and Weiß obtained an in-place sorting method
that uses on the average a close to minimal number of comparisons of $n\lg n - 1.3999n +o(n)$.

In a recent follow-up manuscript~\cite{EdelkampWeiss2018}, Edelkamp and Weiß
investigated the practical performance of QuickXsort and found that
a tuned median-of-3 QuickMergesort variant indeed outperformed the C++
library Quicksort. 
They also derive an upper bound for the average costs of their algorithm
using an inductive proof; their bound is not tight.

\section{Preliminaries}
\label{sec:prelims}

A comprehensive list of used notation is given in \wref{app:notation};
we mention the most important here.
We use Iverson's bracket $[\mathit{stmt}]$ to mean $1$ if $\mathit{stmt}$ is true and $0$ otherwise.
$\Prob{E}$ denotes the probability of event $E$, $\E{X}$ the expectation
of random variable $X$. We write $X\eqdist Y$ to denote equality in distribution.

We heavily use the beta distribution: For $\alpha,\beta \in \R_{>0}$,
$X\eqdist\betadist(\alpha,\beta)$ if $X$ admits the density $f_X(z) = z^{\alpha-1}(1-z)^{\beta-1} / \BetaFun(\alpha,\beta)$
where $\BetaFun(\alpha,\beta) = \int_0^1 z^{\alpha-1}(1-z)^{\beta-1} \, dz$ is the beta function.
Moreover, we use the beta-binomial distribution, which is a conditional binomial distribution 
with the success probability being a beta-distributed random variable.
If $X\eqdist \betaBinomial(n,\alpha,\beta)$ then
$\Prob{X=i} = \binom ni \BetaFun(\alpha+i,\beta+(n-i)) / \BetaFun(\alpha,\beta)$.
For a collection of its properties see~\cite{Wild2016}, Section~2.4.7;
one property that we use here is a local limit law showing that the normalized beta-binomial distribution
converges to the beta distribution. It is reproduced as \wref{lem:beta-binomial-convergence-to-beta}
in the appendix.

For solving recurrences, we build upon Roura's master theorems~\cite{Roura2001}.
The relevant continuous master theorem is restated in the appendix (\wref{thm:CMT}).

\section{QuickXsort}
\label{sec:quickxsort}

Let X be a sorting method that requires buffer space for storing
at most $\lfloor \alpha n\rfloor$ elements (for $\alpha \in [0,1]$) 
to sort $n$ elements. The buffer may only be accessed by swaps 
so that once X has finished its work, the buffer contains the same elements as before, 
but in arbitrary order.
Indeed, we will assume that X does not compare any buffer contents;
then QuickXsort preserves randomness: 
if the original input is a random permutation, so will be the segments after partitioning
and so will be the buffer after X has terminated.%
\footnote{%
	We assume in this paper throughout that the input contains pairwise distinct elements.
}

We can then combine%
\footnote{%
	Depending on details of X, further precautions might have to be taken,
	\eg, in QuickHeapsort~\cite{CantoneCincotti2002}.
	We assume here that those have already been taken care of and solely 
	focus on the analysis of QuickXsort.%
} 
X with Quicksort as follows:
We first randomly choose a pivot and partition the input around that pivot.
This results in two contiguous segments containing the $J_1$ elements that are smaller than the pivot
and the $J_2$ elements that are larger than the pivot, respectively.
We exclude the space for the pivot, so $J_1+J_2 = n-1$;
note that since the rank of the pivot is random, so are the segment sizes $J_1$ and $J_2$.
We then sort one segment by X using the other segment as a buffer,
and afterwards sort the buffer segment recursively by QuickXsort.

To guarantee a sufficiently large buffer for X when it sorts $J_r$ ($r=1$ or $2$), 
we must make sure that $J_{3-r} \ge \alpha J_r$.
In case both segments could be sorted by X, we use the larger one.
The motivation behind this is that we expect an advantage from reducing the
subproblem size for the recursive call as much as possible.

We consider the practically relevant version of QuickXsort,
where we use as pivot the median of a sample of $k=2t+1$ elements,
where $t\in\N_0$ is constant \wrt $n$.
\ifconf{}{We think of $t$ as a design parameter of the algorithm that we have to choose.}
Setting $t=0$ corresponds to selecting pivots uniformly at random.

\subsection{Recurrence for Expected Costs}

Let $c(n)$ be the expected number of comparisons in QuickXsort on arrays of size $n$
and $x(n)$ be (an upper bound for) the expected number of comparisons in X.
We will assume that $x(n)$ fulfills
\begin{align*}
		x(n)
	&\wwrel=
		a n \lg n +b n \wbin\pm \Oh(n^{1-\epsilon}) 
		,\qquad(n\to\infty),
\end{align*}
for constants $a$, $b$ and $\epsilon\in(0,1]$.

For $\alpha<1$, we obtain two cases:
When the split induced by the pivot is ``uneven''~-- 
namely when $\min\{J_1,J_2\} < \alpha\max\{J_1,J_2\}$, \ie, 
$\max\{J_1,J_2\} > \frac{n-1}{1+\alpha}$~-- 
the smaller
segment is not large enough to be used as buffer.
Then we can only assign the large segment as a buffer and run X on the \emph{smaller} segment.
If however the split is about ``even'', \ie, both segments are $\le \tfrac{n-1}{1+\alpha}$
we can sort the \emph{larger} of the two segments by X.
These cases also show up in the recurrence of costs:
\begin{align*}
		c(n) 
	&\wwrel=
		b(n) \ge 0,
	\qquad(n\le k)
\\
		c(n)
	&\wwrel=
			\begin{aligned}[t]
			(n-k) + b(k)  
				&\bin+ \E*{[J_1,J_2 \le \tfrac1{1+\alpha} (n-1) ][J_1 > J_2](x(J_1) + c(J_2))}\\
				&\bin+ \E*{[J_1,J_2 \le \tfrac1{1+\alpha} (n-1) ][J_1 \le J_2](x(J_2) + c(J_1))}\\
				&\bin+ \E*{[J_2 > \tfrac1{1+\alpha}(n-1)](x(J_1) + c(J_2))}\\
				&\bin+ \E*{[J_1 > \tfrac1{1+\alpha}(n-1)](x(J_2) + c(J_1))}
			&\qquad (n \ge 2)
			\end{aligned}
\\	&\wwrel=
		\sum_{r=1}^2\E{A_r(J_r) c(J_r)}  + t(n)
	\numberthis\label{eq:recurrence-cn-E}
\ifconf{\qquad\text{where}\\}{\shortintertext{where}}
		A_1(J)
	&\wwrel=
		[J,J' \le \tfrac1{1+\alpha} (n-1) ]\cdot [J \le J'] \bin+ [J > \tfrac1{1+\alpha}(n-1)]
	\quad\text{with }J' = (n-1)-J
\\
		A_2(J)
	&\wwrel=
		[J,J' \le \tfrac1{1+\alpha} (n-1) ]\cdot [J < J'] \bin+ [J > \tfrac1{1+\alpha}(n-1)]
\\
		t(n)
	&\wwrel=
		(n-1) 
			\bin+ \E*{A_2(J_2) x(J_1)}
			\bin+ \E*{A_1(J_1) x(J_2)}
\label{eq:definition-toll}
\end{align*}

The expectation here is taken over the choice for the random pivot, \ie, over the
segment sizes $J_1$ resp.\ $J_2$.
Note that we use both $J_1$ and $J_2$ to express the conditions in a convenient form,
but actually either one is fully determined by the other via $J_1 + J_2 = n-1$.
Note how $A_1$ and $A_2$ change roles in recursive calls and toll functions, 
since we always sort one segment recursively and the other segment by X.

\ifconf{}{
	The base cases $b(n)$ are the costs to sort inputs that are too small to sample $k$ elements.
	A practical choice is be to switch to Insertionsort for these, which is also used for sorting the samples.
	Unlike for Quicksort itself, $b(n)$ only influences the \emph{logarithmic term} of costs
	(for constant $k$).
	For our asymptotic transfer theorem, we only assume $b(n)\ge 0$, the actual values are immaterial.
}

\paragraph{Distribution of Subproblem Sizes}
If pivots are chosen as the median of a random sample of size $k = 2t+1$,
the subproblem sizes have the same distribution, $J_1\eqdist J_2$.
Without pivot sampling, we have $J_1 \eqdist \uniform[0..n-1]$, a discrete uniform distribution.
If we choose pivots as medians of a sample of $k=2t+1$ elements, the value for $J_1$ 
consists of two summands: $J_1 = t + I_1$. The first summand, $t$, 
accounts for the part of the sample that is smaller than
the pivot. Those $t$ elements do not take part in the partitioning round (but they have to be
included in the subproblem). $I_1$ is the number of elements that turned out to be
smaller than the pivot during partitioning.

This latter number $I_1$ is random, and its distribution is $I_1\eqdist \betaBinomial(n-k,t+1,t+1)$,
a so-called \textsl{beta-binomial distribution}.
The connection to the beta distribution is best seen by assuming $n$ independent and uniformly in $(0,1)$
distributed reals as input. They are almost surely pairwise distinct and their relative
ranking is equivalent to a random permutation of~$[n]$, so this assumption is \withoutlossofgenerality
for our analysis.
Then, the \emph{value} $P$ of the pivot in the first partitioning step has a 
$\betadist(t+1,t+1)$ distribution \emph{by definition}.
\emph{Conditional} on that \emph{value} $P=p$, $I_1 \eqdist \binomial(n-k,p)$ has a binomial distribution;
the resulting mixture is the so-called beta-binomial distribution.

For $t=0$, \ie, no sampling, we have $t+\betaBinomial(n-k,t+1,t+1) = \betaBinomial(n-1,1,1)$,
so we recover the uniform case $\uniform[0..n-1]$.

\section{The Transfer Theorem}
\label{sec:main-result}

We now state the main result of the paper: an asymptotic approximation for $c(n)$.

\begin{theorem}[Total Cost of QuickXsort]
\label{thm:cn}
	The expected number of comparisons needed to sort a 
	random permutation with QuickXsort using median-of-$k$ pivots, $k=2t+1$, 
	and a sorting method X that needs a buffer of $\lfloor \alpha n\rfloor$ elements 
	for some constant $\alpha\in [0,1]$ to sort $n$ elements
	and requires on average $x(n)=a n\lg n +bn \pm \Oh(n^{1-\epsilon})$ comparisons to do so 
	as $n\to\infty$ for some $\epsilon\in(0,1]$ is
	\begin{align*}
			c(n)
		&\wwrel=
			a n \lg n + \biggl(
				\frac 1H - a\cdot\frac {\harm{k+1}-\harm{t+1}}{H \ln 2} +b \biggr) \cdot n
			\wbin\pm\Oh(n^{1-\epsilon} + \log n),
	\ifconf{\\\text{where\;\;}}{\shortintertext{where}}
			H
		&\wwrel=  I_{0,\frac\alpha{1+\alpha}}(t+2,t+1)
			\bin+ I_{\frac12,\frac1{1+\alpha}}(t+2,t+1)
	\end{align*}
	is the expected relative subproblem size that is sorted by X.
\end{theorem}
Here $I_{x,y}(\alpha,\beta)$ is the regularized incomplete beta function
\begin{align*}
		I_{x,y}(\alpha,\beta)
	&\wwrel=
		\int_x^y \frac{z^{\alpha-1}(1-z)^{\beta-1}}{\BetaFun(\alpha,\beta)} \, dz
		,\qquad (\alpha,\beta\in\R_+, 0\le x\le y\le 1).
\end{align*}

We prove \wref{thm:cn} in \wref[Sections]{sec:solution-leading-term}
and~\ref{sec:solution-linear-term}.
To simplify the presentation, we will restrict ourselves to a stereotypical
algorithm for X and its value $\alpha=\frac12$;
the given arguments, however, immediately extend to the general statement above.

\section{QuickMergesort}
\label{sec:quickmergesort}

A natural candidate for X is Mergesort: It is comparison-optimal up to the linear term
(and quite close to optimal in the linear term),
and needs a $\Theta(n)$-element-size buffer for practical implementations of merging.%
\footnote{%
	Merging can be done in place using more advanced tricks (see, \eg, \cite{MannilaUkkonen1984}), 
	but those tend not to be competitive
	in terms of running time with other sorting methods.
	By changing the global structure, a pure in-place Mergesort variant~\cite{KatajainenPasanenJukka1996} 
	can be achieved using part of the input as a buffer (as in QuickMergesort) 
	at the expense of occasionally having to merge runs of very different lengths.
}

To be usable in QuickXsort, we use a swap-based merge procedure as given in \wref{alg:merge}.
\begin{algorithm}[tbh]
	\begin{codebox}
	\Procname{$\proc{Merge}(A[\ell..r],m, B[b..e])$}
	\zi \Comment Merges runs $A[\ell,m-1]$ and $A[m..r]$ in-place into $A[l..r]$ using scratch space $B[b..e]$
	\li $n_1\gets r-\ell+1$; \; $n_2\gets r-\ell+1$
	\zi \Comment Assumes $A[\ell,m-1]$ and $A[m..r]$ are sorted, $n_1 \le n_2$ and $n_1 \le e-b+1$.
	\li \kw{for} $i = 0,\ldots,n_1-1$ \kw{do}
			$\proc{Swap}(A[\ell+i], B[b+i])$ \kw{end for}
	\li $i_1 \gets b$; \; $i_2 \gets m$; $o \gets \ell$
	\li \While $i_1 < b + n_1$ and $i_2 \le r$
	\Do
		\li \kw{if} $B[i_1] \le A[i_2]$ \kw{then}
		\>\>\>\>\>
			$\proc{Swap}(A[o],B[i_2])$; \; $o\gets o+1$; \; $i_1\gets i_1+1$
		\li \kw{else}
		\>\>\>\>\>
			$\proc{Swap}(A[o],A[i_1])$; \; $o\gets o+1$; \; $i_2\gets i_2+1$
		\; \kw{end if}
	\End
	\li \kw{while} $i_1 < b + n_1$ \kw{do}
			$\proc{Swap}(A[o],B[i_2])$; \; $o\gets o+1$; \; $i_1\gets i_1+1$ 
		\; \kw{end while}
	\end{codebox}
	\caption{%
		A simple merging procedure that uses the buffer only by swaps.
		We move the first run $A[\ell..r]$ into the buffer $B[b..b+n_1-1]$ 
		and then merge it with the second run $A[m..r]$ (still stored
		in the original array) into the empty slot left by the first run. 
		By the time this first half is filled, we either have consumed enough of the second run 
		to have space to grow the merged result, or the merging was trivial, \ie, 
		all elements in the first run were smaller.%
	}
	\label{alg:merge}
\end{algorithm}
Note that it suffices to move the \emph{smaller} of the two runs to a buffer;
we use a symmetric version of \wref{alg:merge} when the second run is shorter.
Using classical top-down or bottom-up Mergesort as described in 
any algorithms textbook (\eg~\cite{SedgewickWayne2011}),
we thus get along with $\alpha=\frac12$.

\subsection{Average Case of Mergesort}

The average number of comparisons for Mergesort has the same~-- optimal~-- leading
term $n \lg n$ as in the worst and best case;
and this is true for both the top-down and bottom-up variants.
The coefficient of the \emph{linear} term of the asymptotic expansion, though, is not a constant,
but a bounded periodic function with period $\lg n$, and the functions differ for best, worst,
and average case and the variants of 
Mergesort~\cite{SedgewickFlajolet2013,FlajoletGolin1994,PannyProdinger1995,Hwang1996,Hwang1998}.

In this paper, we will confine ourselves to an \emph{upper bound} for the average case 
$x(n) = a n \lg n +b n \wbin\pm \Oh(n^{1-\epsilon}) $ with \emph{constant} $b$ valid for all $n$,
so we will set $b$ to the supremum of the periodic function.
We leave the interesting challenge open to trace the precise behavior of the 
fluctuations through the recurrence,
where Mergesort is used on a logarithmic number of subproblems with random sizes.

We use the following upper bounds for \underline top-\underline down~\cite{Hwang1998} 
and \underline bottom-\underline up~\cite{PannyProdinger1995} Mergesort%
\footnote{%
	Edelkamp and Weiß~\cite{EdelkampWeiss2014} use 
	$x(n) = n \lg n - 1.26n \wbin\pm o(n)$;
	Knuth~\cite[5.2.4--13]{Knuth1998} derived this formula for $n$ a power of $2$
	(a general analysis is sketched, but no closed result for general $n$ is given).
	Flajolet and Golin~\cite{FlajoletGolin1994} and Hwang~\cite{Hwang1998} continued 
	the analysis in more detail; 
	they find that the average number of comparisons is $n \lg n -(1.25 \pm 0.01) n \pm O(1)$,
	where the linear term oscillates in the given range. 
}
\begin{align}
		x_\mathrm{td}(n)
&	\wwrel=
		n \lg n - 1.24n + 2  \qquad \text{and}
\ifconf{&}{\\}
		x_\mathrm{bu}(n)
&	\wwrel=
		n \lg n - 0.26n \wbin\pm \Oh(1) .
\end{align}

\section{Solving the Recurrence: Leading Term}
\label{sec:solution-leading-term}

We start with \weqref{eq:recurrence-cn-E}. Since $\alpha = \frac12$ for our Mergesort, 
we have $\frac\alpha{1+\alpha} = \frac13$ and $\frac1{1+\alpha} = \frac 23$.
(The following arguments are valid for general $\alpha$, including the extreme case $\alpha=1$,
but in an attempt to de-clutter the presentation, we stick to $\alpha=\frac12$ here.)
We rewrite $A_1(J_1)$ and $A_2(J_2)$ explicitly in terms of the \emph{relative subproblem size}:
\begin{align*}
		A_1(J_1)
	&\wwrel=
		\biggl[\frac{J_1}{n-1} \rel\in \Bigl[\frac13,\frac12\like{\Bigr)}{\Bigr]}\cup \Bigl(\frac23,1\Bigr]\biggr] ,
\ifconf{&}{\\}
		A_2(J_2)
	&\wwrel=
		\biggl[\frac{J_2}{n-1} \rel\in \Bigl[\frac13,\frac12\Bigr)\cup \Bigl(\frac23,1\Bigr]\biggr] .
\end{align*}
Graphically, if we view $J_1/(n-1)$ as a point in the unit interval,
the following picture shows which subproblem is sorted recursively;
(the other subproblem is sorted by Mergesort).
\vspace{-1ex}
\begin{center}\small
\begin{tikzpicture}[yscale=7,xscale=10]
	\draw (0,0) -- (1,0);
	\draw[{Bracket[]}-{Parenthesis[]},thick] (0,0) -- node[below] {$A_2=1$} (1/3,0) ;
	\draw[{Bracket[]}-{Bracket[]},thick] (1/3,0) -- node[above] {$A_1=1$} (1/2,0) ;
	\draw[{Parenthesis[]}-{Bracket[]},thick] (1/2,0) -- node[below] {$A_2=1$} (2/3,0) ;
	\draw[{Parenthesis[]}-{Bracket[]},thick] (2/3,0) -- node[above] {$A_1=1$} (1,0) ;
	\foreach \l in {0,1/3,1/2,2/3,1} {\node at (\l,-0.075) {$\l$} ;}
\end{tikzpicture}
\end{center}
Obviously, we have $A_1 + A_2 = 1$ for any choice of $J_1$, 
which corresponds to having exactly one recursive call
in QuickMergesort.
\subsection{The Shape Function}
The expectations $\E{A_r(J_r) c(J_r)}$ in \weqref{eq:recurrence-cn-E} 
are actually finite sums over the values $0,\ldots, n-1$ that $J \ce J_1$ can attain.
Recall that $J_2 = n-1 - J_1$ and
$A_1(J_1) + A_2(J_2) = 1$ for any value of $J$.
With $J = J_1 \eqdist J_2$, we find
\begin{align*}
		\sum_{r=1}^2\E{A_r(J_r) c(J_r)}
	&\wwrel=
\ifconf{}{
		{}\bin\ppe
		\E[\Bigg]{\biggl[\frac{J}{n-1} \rel\in \Bigl[\frac13,\frac12\like{\Bigr)}{\Bigr]}\cup \Bigl(\frac23,1\Bigr]\biggr] \cdot c(J)}
\\*	&\wwrel\ppe{}	
		\bin+ 
		\E[\Bigg]{\biggl[\frac{J}{n-1} \rel\in \Bigl[\frac13,\frac12\Bigr)\cup \Bigl(\frac23,1\Bigr]\biggr] \cdot c(J)}
\\	&\wwrel=
}
		\sum_{j=0}^{n-1} w_{n,j} \cdot c(j)
		,\qquad \text{where}
\\[1ex]			
		w_{n,j}
	&\wwrel={}
		\phantom{{}+{}}\Prob{J=j} \cdot 
			\Bigl[\tfrac{j}{n-1} \in [\tfrac13,\tfrac12\like(]\cup (\tfrac23,1]\Bigr]
\\*	&\wwrel\ppe{}
		+\Prob{J=j} \cdot 
			\Bigl[\tfrac{j}{n-1} \in [\tfrac13,\tfrac12)\cup (\tfrac23,1]\Bigr]
\\[.5ex]	&\wwrel= 
		\begin{dcases}
			2 \cdot \Prob{J = j} & \text{if } \tfrac{j}{n-1} \in [\tfrac13,\tfrac12)\cup (\tfrac23,1] \\
			1 \cdot \Prob{J=j} & \text{if } \tfrac{j}{n-1} = \tfrac12 \\
			0	& \text{otherwise}
		\end{dcases}
\end{align*}
We thus have a recurrence of the form required by the Roura's \textsl{continuous master theorem (CMT)}
(see \wref{thm:CMT} in \wref{app:CMT}) with the weights $w_{n,j}$ from above 
(\wref{fig:wnj} shows an example how these weights look like).

\begin{figure}[htb]
\plaincenter{\includegraphics[width=.5\linewidth]{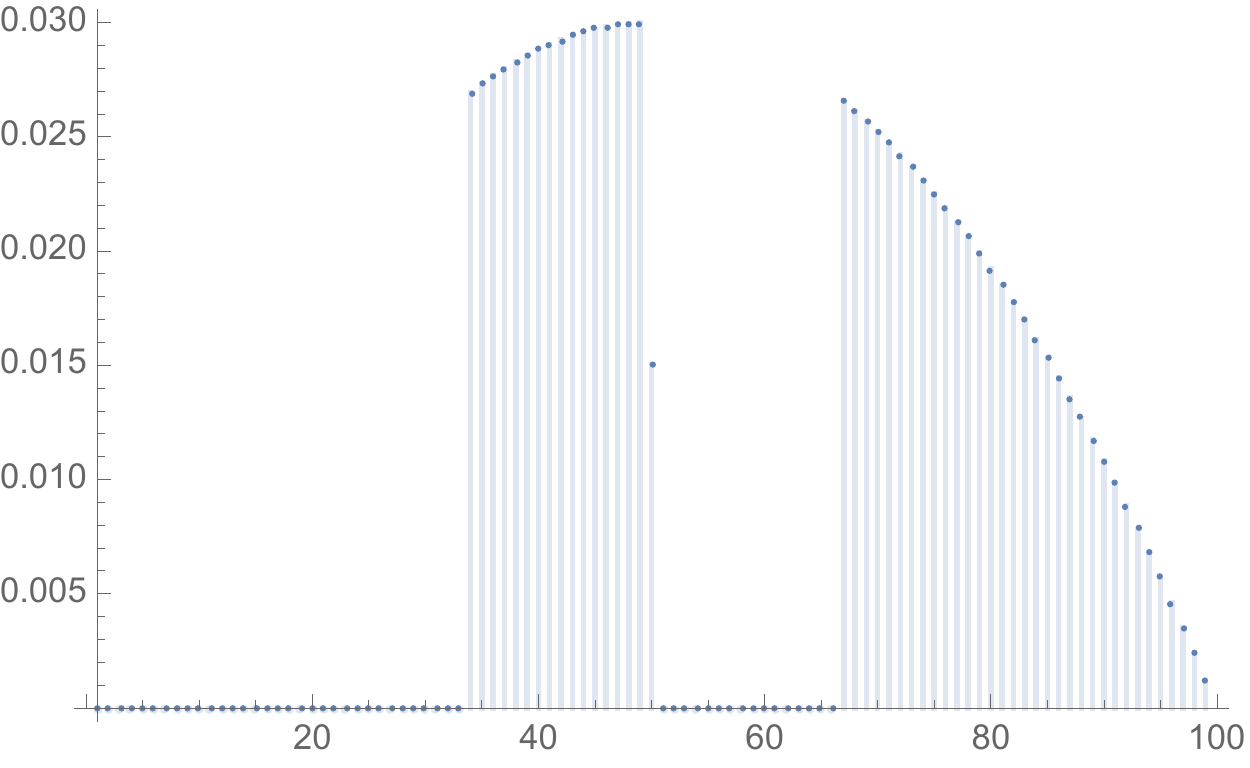}}
\caption{%
	The weights $w_{n,j}$ for $n=101$, $t=1$; note the singular point at $j=50$.
}
\label{fig:wnj}
\end{figure}

It remains to determine $\Prob{J=j}$.
Recall that we choose the pivot as the median of $k=2t+1$ elements for a fixed constant $t\in\N_0$,
and the subproblem size $J$ fulfills $J = t + I$ with $I\eqdist \betaBinomial(n-k,t+1,t+1)$.
So we have for $i\in[0,n-1-t]$ by definition
\begin{align*}
		\Prob{I=i}
	&\wwrel=
		\binom{n-k}{i} \frac{\BetaFun\bigl(i+t+1,(n-k-i))+t+1\bigr)}{\BetaFun(t+1,t+1)}
\\	&\wwrel=
		\binom{n-k}{i} \frac{ (t+1)^{\overline{i}} (t+1)^{\overline{n-k-i}} }{(k+1)^{\overline{n-k}}}
\end{align*}
(For details, see \cite[Section 2.4.7]{Wild2016}.)
Now the local limit law for beta binomials 
(\wref{lem:beta-binomial-convergence-to-beta} in \wref{app:local-limit-beta} says that
the normalized beta binomial $I/n$ converges to a beta variable ``in density'',
and the convergence is uniform.
With the beta density $f_P(z) = z^t(1-z)^t / \BetaFun(t+1,t+1)$, 
we thus find by \wref{lem:beta-binomial-convergence-to-beta} that
\begin{align*}
		\Prob{J=j}
	&\wwrel=
		\Prob{I=j-t}
	\wwrel=
		\frac1n f_P(j/n) \wwbin\pm \Oh(n^{-2}) 
		,\qquad (n\to\infty).
\end{align*}
The shift by the small constant $t$ from $(j-t)/n$ to $j/n$
only changes the function value by $\Oh(n^{-1})$ since $f_P$ is Lipschitz continuous on $[0,1]$.
(Details of that calculation are also given in \cite{Wild2016}, page~208.)

The first step towards applying the CMT is to identify a shape function $w(z)$
that approximates the relative subproblem size probabilities $w(z) \approx n w_{n,\lfloor z n\rfloor}$ 
for large $n$.
With the above observation, a natural choice is
\begin{align*}
\numberthis\label{eq:shape-function}
		w(z)
	&\wwrel=
		2\,\bigl[\tfrac13 < z < \tfrac12 \bin\vee z > \tfrac23\bigr] 
			\frac{z^{t}(1-z)^{t}}{\BetaFun(t+1,t+1)}.
\end{align*}
We show in \wref{app:shape-function-condition} that this is indeed a suitable shape function,
\ie, it fulfills \weqref{eq:CMT-shape-function-condition} from the CMT.

\subsection{Computing the Toll Function}
\label{sec:toll-function-leading-term}

The next step in applying the CMT is a leading-term approximation
of the toll function.
We consider a general function $x(n) = a n \lg n + b n \pm \Oh(n^{1-\epsilon})$ where
the error term holds for any constant $\epsilon >0$ as $n\to\infty$. 
We start with the simple observation that
\begin{align*}
		J \lg J
	&\wwrel=
		J \bigl(\lg (\tfrac Jn) + \lg n\bigr)
\ifconf{}{\\	&}\wwrel=
		n \cdot \Bigl( \tfrac Jn \lg \tfrac Jn + \tfrac Jn \lg n  \Bigr)
\ifconf{}{\\	&}\wwrel=
		\tfrac J n \, n \lg n\bin+ \tfrac Jn \lg \bigl(\tfrac Jn\bigr)\, n.
	\numberthis\label{eq:ex-JlnJ}
\\	&\wwrel=
		\tfrac J n \, n \lg n \wwbin\pm \Oh(n).
	\numberthis\label{eq:ex-JlnJ-On}
\end{align*}
For the leading term of $\E{x(J)}$, we thus only have to compute 
the expectation of $J/n$, which is essentially a relative subproblem size.
In $t(n)$, we also have to deal with the conditionals $A_1(J)$ resp.\
$A_2(J)$, though. By approximating $\frac Jn$ with a
beta distributed variable, the conditionals translate to bounds of an integral. 
Details are given in \wref{lem:E-Jn-int} (see \wref{app:approx-by-beta-integrals}). This yields
\begin{align*}
		t(n)
	&\wwrel=
		n-1 + \E*{A_2(J_2) x(J_1)} + \E*{A_1(J_1) x(J_2)}
\\	&\wwrel=
		a\, \E*{A_2(J_2) J_1 \lg J_1} + a\, \E*{A_1(J_1) J_2 \lg J_2)} \wwbin\pm \Oh(n)
\\	&\wwrel{\eqwithref[r]{lem:E-xy-Jn}}
		2a\cdot \frac{t+1}{2t+2} \cdot\Bigl(
			I_{0,\frac13}(t+2,t+1) + I_{\frac12,\frac23}(t+2,t+1)
		\Bigr) \cdot n\lg n \wwbin\pm \Oh(n)
\\	&\wwrel=
		\underbrace{
			a\Bigl(
				I_{0,\frac13}(t+2,t+1) + I_{\frac12,\frac23}(t+2,t+1)
			\Bigr) 
		} _ {\bar a} {}
		\cdot n\lg n \wwbin\pm \Oh(n)
		,\qquad(n\to\infty).
\numberthis\label{eq:tn-leading-term}
\end{align*}
Here we use the \textsl{incomplete regularized beta function} 
\begin{align*}
		I_{x,y}(\alpha,\beta)
	&\wwrel=
		\int_x^y \frac{z^{\alpha-1}(1-z)^{\beta-1}}{\BetaFun(\alpha,\beta)} \, dz
		,\qquad (\alpha,\beta\in\R_+, 0\le x\le y\le 1)
\end{align*}
for concise notation. ($I_{x,y}(\alpha,\beta)$ is the probability
that a $\betadist(\alpha,\beta)$ distributed random variable falls into $(x,y)\subset[0,1]$,
and $I_{0,x}(\alpha,\beta)$ is its cumulative distribution function.)

\subsection{Which Case of the CMT?}

We are now ready to apply the CMT (\wref{thm:CMT}).
As shown in \wref{sec:toll-function-leading-term}, our toll function is $\Theta(n \log n)$,
so we have $\alpha=1$ and $\beta=1$.
We hence compute 
\begin{align*}
		H
	&\wwrel=
		1-\int_0^1 z\,w(z) \:dz 
\\	&\wwrel= 
		1 - \int_0^1 2\,\bigl[\tfrac13 < z < \tfrac12 \bin\vee z > \tfrac23\bigr] 
					\frac{z^{t+1}(1-z)^{t}}{\BetaFun(t+1,t+1)} \: dz
\\	&\wwrel= 
		1 - 2\frac{t+1}{k+1}
				\int_0^1 \bigl[\tfrac13 < z < \tfrac12 \bin\vee z > \tfrac23\bigr] 
				\frac{z^{t+1}(1-z)^{t}}{\BetaFun(t+2,t+1)} \: dz
\\	&\wwrel= 
		1 - \Bigl( I_{\frac13,\frac12}(t+2,t+1)
			+ I_{\frac23,1}(t+2,t+1)\Bigr)
\\	&\wwrel= 
		  I_{0,\frac13}(t+2,t+1)
		+ I_{\frac12,\frac23}(t+2,t+1)
\numberthis\label{eq:H}
\end{align*}
For any sampling parameters, we have $H>0$,
so the overall costs satisfy by Case~1 of \wref{thm:CMT}
\begin{align*}\SwapAboveDisplaySkip
		c(n)
	&\wwrel\sim
		\frac{t(n)}{H}
	\wwrel\sim
		\frac{\bar a n \lg n}{H}
		,\qquad (n\to\infty).
\numberthis\label{eq:cn-leading-term}
\end{align*}

\subsection{Cancellations}

Combining \wref[Equations]{eq:tn-leading-term} and~\wref{eq:cn-leading-term}, we find
\ifconf{%
$	
		c(n)
	\sim
		a n \lg n
		,\text{ as }(n\to\infty),
$%
}{%
\begin{align*}
		c(n)
	&\wwrel\sim
		a n \lg n
		,\qquad(n\to\infty);
\end{align*}
}
since $I_{0,\frac13} + I_{\frac13,\frac12} + I_{\frac12,\frac23} + I_{\frac23,1} = 1$.
The leading term of the number of comparisons in QuickXsort is the \emph{same} as in X itself,
regardless of how the pivot elements are chosen!
This is not as surprising as it might first seem.
We are typically sorting a constant fraction of the input by X
and thus only do a logarithmic number of recursive calls on a 
geometrically decreasing number of elements,
so the linear contribution of Quicksort (partitioning and recursion cost) is
dominated by even the first call of X, which has linearithmic cost.
This remains true even if we allow asymmetric sampling, \eg,
by choosing the pivot as the \emph{smallest} (or any other order statistic)
of a random sample.

Edelkamp and Weiß~\cite{EdelkampWeiss2014} give the above result for the case of using
the median of $\sqrt n$ elements, where we effectively have exact medians from
the perspective of analysis.
In this case, the informal reasoning given above is precise, and in fact, 
in this case the same form of cancellations also happen for the linear term
\cite[Thm.\ 1]{EdelkampWeiss2014}. (See also the ``exact ranks'' result in \wref{sec:discussion}.)
We will show in the following that for practical schemes of pivot sampling, \ie, with fixed sample sizes,
these cancellations happen \emph{only for the leadings-term approximation.}
The pivot sampling scheme does affect the linear term significantly;
and to measure the benefit of sampling,
the analysis thus has to continue to the next term of the asymptotic expansion of~$c(n)$.

\paragraph{Relative Subproblem Sizes}
The integral $\int_0^1 z w(z)\,dz$ is precisely the expected relative subproblem size for the recursive
call, whereas for $t(n)$ we are interested in the subproblem that is sorted using X
whose relative size is given by
$\int_0^1 (1-z)w(z)\, dz = 1-\int_0^1 zw(z)\, dz$.
We can thus write $\bar a = a H$.

\begin{figure}[htpb]
	\plaincenter{\includegraphics[width=.5\linewidth]{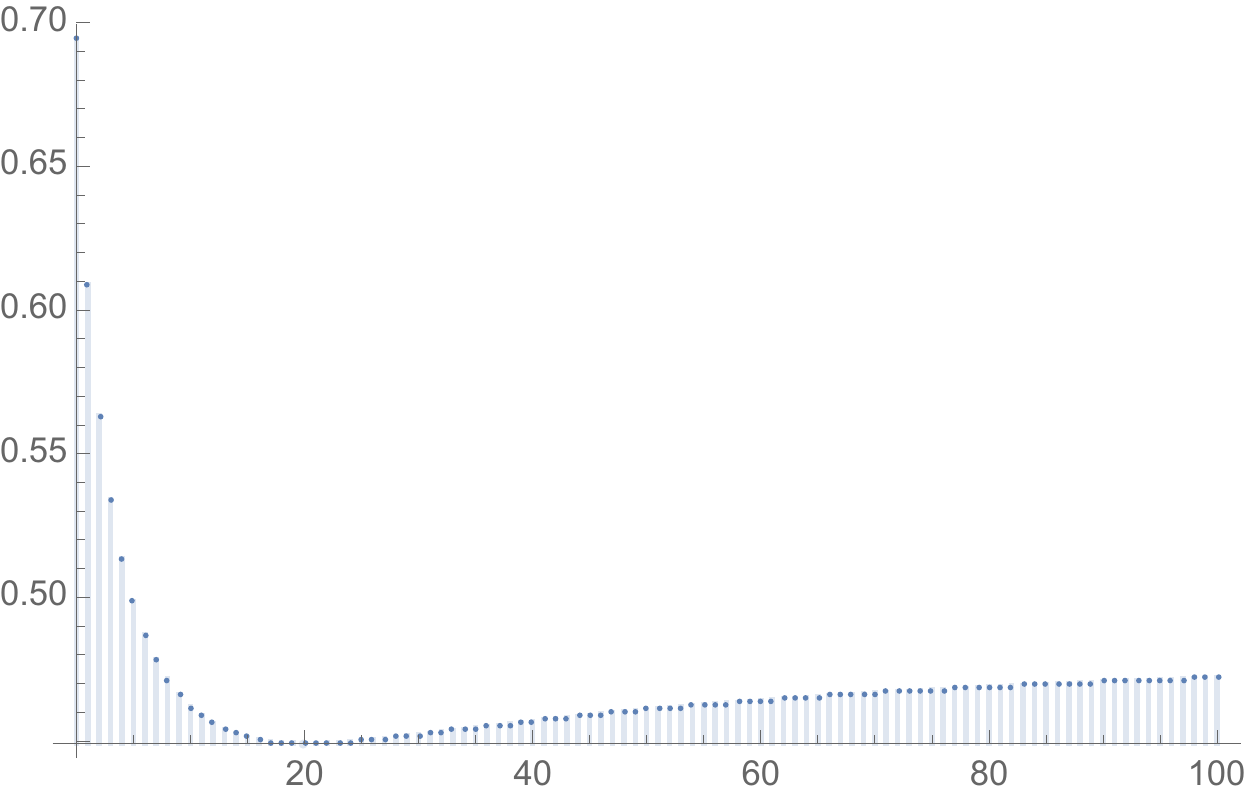}}
	\caption{%
		$\int_0^1 z w(z)\,dz$, the relative recursive subproblem size, as a function of $t$.
	}
	\label{fig:recursive-subproblem-size}
\end{figure}

The quantity $\int_0^1 z w(z)\,dz$, the average relative size of the recursive call
is of independent interest.
While it is intuitively clear that for $t\to\infty$, \ie, the case of exact medians as pivots,
we must have a relative subproblem size of exactly \smash{$\frac12$},
this convergence is not
apparent from the behavior for finite $t$: the mass of the integral $\int_0^1 z w(z)\,dz$
concentrates at $z=\frac12$, a point of discontinuity in $w(z)$.
It is also worthy of note that the expected subproblem size is initially larger than $\frac12$ 
($0.69\overline 4$ for $t=0$), then decreases to $\approx 0.449124$ around $t=20$ and then starts 
to slowly increase again%
\ifconf{.}{ (see \wref{fig:recursive-subproblem-size}).}

\section{Solving the Recurrence: The Linear Term}
\label{sec:solution-linear-term}

Since $c(n) \sim a n \lg n$ for any choice of $t$,
the leading term alone 
does not allow to make distinctions to judge the effect of sampling schemes.
To compute the next term in the asymptotic expansion of $c(n)$,
we consider the values $c'(n) = c(n) - a n\lg n$.
$c'(n)$ has essentially the same recursive structure as $c(n)$, only with a different toll function:
\begin{align*}
		c'(n)
	&\wwrel=
		c(n) - a n\lg n
\\	&\wwrel=
		\sum_{r=1}^2 \E[\big]{A_r(J_r) c(J_r)} - a n \lg n + t(n)
\\	&\wwrel=
		\sum_{r=1}^2 \biggl( 
			\E[\Big]{A_r(J_r) \bigl(c(J_r) - a J_r \lg J_r \bigr)} 
			+ a \,\E[\big]{A_r(J_r) J_r \lg J_r} 
		\biggr)
		\bin- a n \lg n 
\\*	&\wwrel\ppe\quad{}
		\wbin+ 		(n-1) 
					\bin+ \E[\big]{A_2(J_2) \cdot x(J_1)}
					\bin+ \E[\big]{A_1(J_1) \cdot x(J_2)}
\\	&\wwrel=
		\sum_{r=1}^2 \E[\Big]{A_r(J_r) c'(J_r)}
		+(n-1) - a n\lg n
\\*	&\wwrel\ppe\quad{}
		\bin+ a\,\E[\Big]{\bigl( A_1(J_1) +  A_2(J_2) \bigr)J_1 \lg J_1} + b\,\E{A_2(J_2) J_1}
\\*	&\wwrel\ppe\quad{}
		\bin+ a\,\E[\Big]{\bigl( A_2(J_2) +  A_1(J_1) \bigr)J_2 \lg J_2} + b\,\E{A_1(J_1) J_2}
		\wwbin\pm \Oh(n^{1-\epsilon})
\end{align*}
Since $J_1\eqdist J_2$ we can simplify
\begin{align*}
	&\wwrel\ppe	\E[\Big]{\bigl( A_1(J_1) +  A_2(J_2) \bigr)J_1 \lg J_1}
		\bin+ \E[\Big]{\bigl( A_2(J_2) +  A_1(J_1) \bigr)J_2 \lg J_2}
\\	&\wwrel=
		\E[\Big]{\bigl( A_1(J_1) +  A_2(J_2) \bigr)J_1 \lg J_1}
		\bin+ \E[\Big]{\bigl( A_2(J_1) +  A_1(J_2) \bigr)J_1 \lg J_1}
\\	&\wwrel=
		\E[\Big]{J_1 \lg J_1\cdot \Bigl(\bigl( A_1(J_1) + A_1(J_2) \bigr)+\bigl( A_2(J_1) + A_2(J_2) \bigr)\Bigr)}
\\	&\wwrel=
		2\E{J \lg J}
\\	&\wwrel{\eqwithref{eq:ex-JlnJ}}
		2 \,\E{\tfrac Jn} \cdot n\lg n + 2\cdot \tfrac 1{\ln 2} \E{\tfrac Jn \ln \tfrac Jn} \cdot n
\\	&\wwrel{\eqwithref[c]{lem:E-Jn-ln-Jn}}
		n\lg n - \tfrac1{\ln 2}\bigl(\harm{k+1}-\harm{t+1}\bigr) n \wwbin\pm\Oh(n^{1-\epsilon}).
\end{align*}

Plugging this back into our equation for $c'(n)$, we find
\begin{align*}
		c'(n)
	&\wwrel=
		\sum_{r=1}^2 \E[\Big]{A_r(J_r) c'(J_r)}
		+(n-1) - a n\lg n
\\*	&\wwrel\ppe\quad{}
		\bin+ a \,\Bigl( n\lg n - \tfrac1{\ln 2}\bigl(\harm{k+1}-\harm{t+1}\bigr) n \Bigr) 
\\*	&\wwrel\ppe\quad{}
		\bin+ b\,\Bigl(I_{0,\frac13}(t+2,t+1)+I_{\frac12,\frac23}(t+2,t+1)\Bigr) \cdot n
		\wwbin\pm \Oh(n^{1-\epsilon})
\\	&\wwrel=
		\sum_{r=1}^2 \E[\Big]{A_r(J_r) c'(J_r)} \bin+ t'(n)
\shortintertext{where}
		t'(n)
	&\wwrel=
		b' n \wbin\pm \Oh(n^{1-\epsilon})
\ifconf{\quad\text{with}\quad}{\\}
	b'
	\ifconf{}{&}\wwrel=
			1 - \tfrac a{\ln 2}\bigl(\harm{k+1}-\harm{t+1}\bigr)
			+b \cdot H
\end{align*}
Apart from the smaller toll function $t'(n)$, this recurrence has the very same shape
as the original recurrence for $c(n)$;
in particular, we obtain the same shape function $w(z)$ and the same $H > 0$
and obtain
\begin{align*}
		c'(n)
	&\wwrel\sim
		\frac{t'(n)}{H}
	\wwrel\sim
		\frac{b'n}{H}.
\end{align*}

\subsection{Error Bound}

Since our toll function is not given precisely, but only up to an error term
$\Oh(n^{1-\epsilon})$ for a given fixed $\epsilon\in(0,1]$,
we also have to estimate the overall influence of this term.
For that we consider the recurrence for $c(n)$ again, but replace $t(n)$ (entirely) by
$C \cdot n^{1-\epsilon}$.
If $\epsilon > 0$, $\int_0^1 z^{1-\epsilon} w(z)\, dz < \int_0^1 w(z)\, dz = 1$,
so we still find $H>0$ and apply case~1 of the CMT\@. The overall contribution of the 
error term is then $\Oh(n^{1-\epsilon})$.
For $\epsilon=0$, $H=0$ and case~2 applies, giving an overall error term of $\Oh(\log n)$.

\medskip\noindent
This completes the proof of \wref{thm:cn}.

\section{Discussion}
\label{sec:discussion}

Since all our choices for X are leading-term optimal, so will QuickXsort be. 
We can thus fix $a=1$ in \wref{thm:cn}; only $b$ (and the allowable $\alpha$) still depend on X.
We then basically find that going from X to QuickXsort adds a ``penalty'' $q$
in the linear term that depends only on the sampling size (and $\alpha$), but not on X.
\wref{tab:q} shows that this penalty is $\approx n$ without sampling, but can be reduced
drastically when choosing pivots from a sample of $3$ or $5$ elements.
(Note that the overall costs for pivot sampling are $\Oh(\log n)$ for constant $t$.)

\begin{table}[bht]
	\begin{center}
	\begin{tabular}{ r *{6}{c} }
	\toprule
		                 & $t=0$    &  $t=1$   &  $t=2$   &  $t=3$   &  $t=10$   & $t\to\infty$ \\
	\midrule
		$\alpha=1$       & $1.1146$ & $0.5070$ & $0.3210$ & $0.2328$ & $0.07705$ &     $0$      \\
		$\alpha=\frac12$ & $0.9120$ & $0.4050$ & $0.2526$ & $0.1815$ & $0.05956$ &     $0$      \\
	\bottomrule
	\end{tabular}
	\end{center}
	\caption{%
		QuickXsort penalty.
		QuickXsort with $x(n) = n\lg n + bn$ yields
		$c(n) = n\lg n + (q +b) n$ where $q$, the QuickXsort penalty, is given in the table.
	}
	\label{tab:q}
\end{table}

As we increase the sample size, we converge to the situation studied by Edelkamp and Weiß
using median-of-$\sqrt n$,
where no linear-term penalty is left~\cite{EdelkampWeiss2014}.
Given that $q$ is less than $0.08$ already for a sample of $21$ elements, these large sample
versions are mostly of theoretical interest.
It is noteworthy that the improvement from no sampling to median-of-3 yields a 
reduction of $q$ by more than $50\%$, which is much more than
its effect on Quicksort itself 
(where it reduces the leading term of costs by 15\,\% from $2n\ln n$ to $\frac{12}7 n\ln n$).

We now apply our transfer theorem to the two most well-studied choices for X, Heapsort and Mergesort,
and compare the results to analyses and measured comparison counts from previous work.
The results confirm that solving the QuickXsort recurrence exactly yields much more accurate
predictions for the overall number of comparisons than previous bounds that circumvented
this.

\subsection{QuickHeapsort}

The basic external Heapsort of Cantone and Cincotti~\cite{CantoneCincotti2002}
always traverses one path in the heap from root to bottom and does one comparison for each edge followed,
\ie, $\lfloor \lg n \rfloor$ or $\lfloor \lg n \rfloor-1$ many per deleteMax.
By counting how many leaves we have on each level, Diekert and Weiß found~\cite[Eq.\ 1]{DiekertWeiss2016}
\[
		n \bigl(\lfloor \lg n\rfloor -1\bigr) + 2 \bigl(n-2^{\lfloor \lg n\rfloor}\bigr) \bin\pm\Oh(\log n) 
	\wwrel\le 
		n\lg n -0.913929n \bin\pm\Oh(\log n)
\]
comparisons for the sort-down phase.
(The constant of the linear term is $1-\frac1{\ln2}-\lg(2\ln 2)$, the supremum of the periodic function at the linear term). 
Using the classical heap construction method adds on average $1.8813726n$ comparisons~\cite{Doberkat1984}, so here
\begin{align*}
		x(n)
	&\wwrel=
		n\lg n + 0.967444 n \wwbin\pm \Oh(n^\epsilon) 
\ifconf{\qquad\text{for any $\epsilon>0$.}}{}
\end{align*}
\ifconf{}{for any $\epsilon>0$.}%


Both~\cite{CantoneCincotti2002} and~\cite{DiekertWeiss2016} report averaged comparison counts
from running time experiments. We compare them in \wref{tab:cmps-qhs}
against the estimates from our result and previous analyses.
While the approximation is not very accurate for $n=100$ (for all analyses),
for larger $n$, our estimate is correct up to the first three digits, whereas 
previous upper bounds have almost one order of magnitude bigger errors.
Note that it is expected for our bound to still be on the conservative side
since we used the supremum of the periodic linear term for Heapsort.

\begin{table}[htpb]
	\begin{center}
	\smaller
	\def\na{\multicolumn1c{---}}
	\begin{tabular}{ r *{4}{ >{$} r <{$} } }
	\toprule
		\multicolumn1c{\textbf{Instance}}                            & \multicolumn1c{\textbf{observed}} & \multicolumn1c{\textbf{W}} & \multicolumn1c{\textbf{CC}} & \multicolumn1c{\textbf{DW}} \\
	\midrule
		Fig.\,4~\cite{CantoneCincotti2002}, $n=10^2$, $k=1$ & 806                      & +67               & +158               & +156               \\
		Fig.\,4~\cite{CantoneCincotti2002}, $n=10^2$, $k=3$ & 714                      & +98               & \na                & +168               \\
		Fig.\,4~\cite{CantoneCincotti2002}, $n=10^5$, $k=1$ & 1\,869\,769              & -600              & +90\,795           & +88\,795           \\
		Fig.\,4~\cite{CantoneCincotti2002}, $n=10^5$, $k=3$ & 1\,799\,240              & +9\,165           & \na                & +79\,324           \\
		Fig.\,4~\cite{CantoneCincotti2002}, $n=10^6$, $k=1$ & 21\,891\,874             & +121\,748         & +1\,035\,695       & +1\,015\,695       \\
		Fig.\,4~\cite{CantoneCincotti2002}, $n=10^6$, $k=3$ & 21\,355\,988             & +49\,994          & \na                & +751\,581          \\[1ex]
		Tab.\,2~\cite{DiekertWeiss2016}, $n=10^4$, $k=1$    & 152\,573                 & +1\,125           & +10\,264           & +10\,064           \\
		Tab.\,2~\cite{DiekertWeiss2016}, $n=10^4$, $k=3$    & 146\,485                 & +1\,136           & \na                & +8\,152            \\
		Tab.\,2~\cite{DiekertWeiss2016}, $n=10^6$, $k=1$    & 21\,975\,912             & +37\,710          & +951\,657          & +931\,657          \\
		Tab.\,2~\cite{DiekertWeiss2016}, $n=10^6$, $k=3$    & 21\,327\,478             & +78\,504          & \na                & +780\,091          \\
	\bottomrule
	\end{tabular}%
	\end{center}
	\caption{%
		Comparison of estimates from this paper (W),
		Theorem~6 of~\cite{CantoneCincotti2002} (CC) and Theorem~1 of~\cite{DiekertWeiss2016} (DW);
		shown is the difference between the estimate and the observed average.
	}
	\label{tab:cmps-qhs}
	\vspace{-2ex}
\end{table}

\subsection{QuickMergesort}

For QuickMergesort, Edelkamp and Weiß~\cite[Fig.\ 4]{EdelkampWeiss2014} report measured average 
comparison counts for a median-of-3 version using top-down Mergesort: 
the linear term is shown to be between $-0.8n$ and $-0.9n$.
In a recent manuscript~\cite{EdelkampWeiss2018}, 
they also analytically consider the simplified median-of-3 QuickMergesort which
\emph{always} sorts the \emph{smaller} segment by Mergesort (\ie, $\alpha=1$).
It uses $n\lg n-0.7330n+o(n)$ comparisons on average (using $b=-1.24$).
They use this as a (conservative) upper bound for the original QuickMergesort.

Our transfer theorem shows that this bound is off by roughly $0.1n$:
median-of-3 QuickMergesort uses at most $c(n) = n\lg n - 0.8350 n \pm \Oh(\log n)$ 
comparisons on average.
Going to median-of-5 reduces the linear term to $-0.9874n$, 
which is better than the worst-case for top-down Mergesort for most $n$.


\oldparagraph{Skewed Pivots for Mergesort?}

For Mergesort with $\alpha=\frac12$ the largest fraction of elements we can
sort by Mergesort in one step is $\frac23$;
this suggests that using a slightly skewed pivot might be beneficial since it 
will increase the subproblem size for Mergesort and decrease the size for recursive calls.
Indeed, Edelkamp and Weiß allude to this variation:
	\textit{``With about
	15\% the time gap, however, is not overly big, and may be bridged with additional 
	efforts like skewed pivots and refined partitioning.''}
(the statement appears in the arXiv version of~\cite{EdelkampWeiss2014}, 
\href{https://arxiv.org/abs/1307.3033}{\texttt{arxiv.org/abs/1307.3033}}).
And the above mentioned StackExchange post actually chooses pivots as the second tertile.

Our analysis above can be extended to skewed sampling schemes
(omitted due to space constraints),
but to illustrate this point it suffices to
pay a short visit to ``wishful-thinking land'' and assume that we can get exact quantiles for free.
We can show (\eg, with Roura's discrete master theorem~\cite{Roura2001})
that if we always pick the exact $\rho$-quantile of the input, for $\rho\in(0,1)$,
the overall costs are
\[
	c_\rho(n) \wwrel=
	\begin{dcases*}
	n\lg n + \biggl(\frac{1 + h(\rho)}{1-\rho} + b\biggr)n \wbin\pm\Oh(n^{1-\epsilon})
		& if $\rho \in (\frac13,\frac12)\cup(\frac23,1)$ \\
	n\lg n + \biggl(\frac{1 + h(\rho)}{\rho} + b\biggr)n \wbin\pm\Oh(n^{1-\epsilon})
		& otherwise
	\end{dcases*}
\]
for $h(x) = x \lg x + (1-x)\lg(1-x)$.
The coefficient of the linear term has a \emph{strict minimum} at $\rho=\frac12$:
Even for $\alpha=\frac12$, the best choice is to use the median of a sample.
(The result is the same for fixed-size samples.)
For QuickMergesort, skewed pivots turn out to be a pessimization,
despite the fact that we sort a larger part by Mergesort.
A possible explanation is that skewed pivots significantly decrease 
the amount of information we obtain from the comparisons during partitioning,
but do not make partitioning any cheaper.

\subsection{Future Work}

More promising than skewed pivot sampling is the use of \emph{several} pivots.
The resulting MultiwayQuickXsort would be able to sort all but one segment using X
and recurse on only one subproblem.
Here, determining the expected subproblem sizes becomes a challenge,
in particular for $\alpha<1$; we leave this for future work.

We also confined ourselves to the expected number of comparisons here,
but more details about the distribution of costs are possible to obtain.
The variance follows a similar recurrence as the one studied in this paper
and a distributional recurrence for the costs can be given.
The discontinuities in the subproblem sizes add a new facet to these analyses.

Finally, it is a typical phenomenon that constant-factor optimal sorting methods
exhibit periodic linear terms. QuickXsort inherits these fluctuations
but smooths them through the random subproblem sizes.
Explicitly accounting for these effects is another interesting challenge for future work.

\myacknowledgements

\clearpage
\appendix

\section{Notation}
\label{app:notation}

\subsection{Generic Mathematics}
\begin{notations}
\notation{$\N$, $\N_0$, $\Z$, $\R$}
	natural numbers $\N = \{1,2,3,\ldots\}$, 
	$\N_0 = \N \cup \{0\}$,
	integers $\Z = \{\ldots,-2,-1,0,1,2,\ldots\}$,
	real numbers $\R$.
\notation{$\R_{>1}$, $\N_{\ge3}$ etc.}
	restricted sets $X_\mathrm{pred} = \{x\in X : x \text{ fulfills } \mathrm{pred} \}$.
\notation{$0.\overline 3$}
	repeating decimal; $0.\overline3 = 0.333\ldots = \frac13$; \\
	numerals under the line form the repeated part of 
	the decimal number.
\notation{$\ln(n)$, $\lg(n)$}
	natural and binary logarithm; $\ln(n) = \log_e(n)$, $\lg(n) = \log_2(n)$.
\notation{$X$}
	to emphasize that $X$ is a random variable it is Capitalized.
\notation{$[a,b)$}
	real intervals, the end points with round parentheses are excluded, 
	those with square brackets are included.
\notation{$[m..n]$, $[n]$}
	integer intervals, $[m..n] = \{m,m+1,\ldots,n\}$;
	$[n] = [1..n]$.
\notation{$[\text{stmt}]$, $[x=y]$}
	Iverson bracket, $[\text{stmt}] = 1$ if stmt is true, $[\text{stmt}] = 0$ otherwise.
\notation{$\harm n$}
	$n$th harmonic number; $\harm n = \sum_{i=1}^n 1/i$.
\notation{$x \pm y$}
	$x$ with absolute error $|y|$; formally the interval $x \pm y = [x-|y|,x+|y|]$;
	as with $\Oh$-terms, we use one-way equalities $z=x\pm y$ instead of $z \in x \pm y$.
\notation{$\BetaFun(\alpha,\beta)$}
	the beta function, $\BetaFun(\alpha,\beta) = \int_0^1 z^{\alpha-1}(1-z)^{\beta-1}\, dz$
\notation{$I_{x,y}(\alpha,\beta)$}
	the regularized incomplete beta function;
	$I_{x,y}(\alpha,\beta)=
			\int_x^y \frac{z^{\alpha-1}(1-z)^{\beta-1}}{\BetaFun(\alpha,\beta)} \, dz$
			for $\alpha,\beta\in\R_+$, $0\le x\le y\le 1$.
\notation{$a^{\underline b}$, $a^{\overline b}$}
	factorial powers;
	``$a$ to the $b$ falling resp.\ rising.''
\end{notations}

\subsection{Stochastics-related Notation}
\begin{notations}
\notation{$\Prob{E}$, $\Prob{X=x}$}
	probability of an event $E$ resp.\ probability for random variable $X$ to
	attain value $x$.
\notation{{$\E{X}$}}
	expected value of $X$; we write $\E{X\given Y}$ for the conditional expectation
	of $X$ given $Y$, and $\Eover X{f(X)}$ to emphasize that expectation is taken 
	\wrt random variable $X$.
\notation{$X\eqdist Y$}
	equality in distribution; $X$ and $Y$ have the same distribution.
\notation{$\uniform(a,b)$}
	uniformly in $(a,b)\subset\R$ distributed random variable. 
\notation{$\betadist(\alpha,\beta)$}
	Beta distributed random variable with shape parameters $\alpha\in\R_{>0}$ and $\beta\in\R_{>0}$.
\notation{$\binomial(n,p)$}
	binomial distributed random variable with $n\in\N_0$ trials and success probability $p\in[0,1]$.\\
\notation{$\betaBinomial(n,\alpha,\beta)$}
	beta-binomial distributed random variable;
	$n\in\N_0$, $\alpha,\beta\in\R_{>0}$;
\end{notations}

\subsection{Notation for the Algorithm}
\begin{notations}
\notation{$n$}
	length of the input array, \ie, the input size.
\notation{$k$, $t$}
	sample size $k\in \N_{\ge1}$, odd; $k=2t+1$, $t\in\N_0$.
\notation{$x(n)$, $a$, $b$}
	Average costs of X,
	$x(n) = a n \lg n + b n \pm O(n^{1-\epsilon})$.
\notation{$t(n)$, $\bar a$, $\bar b$}
	toll function
	$t(n) = \bar a n \lg n + \bar b n \pm O(n^{1-\epsilon})$
\notation{$J_1$, $J_2$}
	(random) subproblem sizes; $J_1+J_2 = n-1$;
	$J_1 = t + I_1$;
\notation{$I_1$, $I_2$}
	(random) segment sizes in partitioning;
	$I_1\eqdist \betaBinomial(n-k, t+1,t+1)$;
	$I_2 = n-k-I_1$;
	$J_1 = t + I_1$
\end{notations}

\section{The Continuous Master Theorem}
\label{app:CMT}

We restate Roura's CMT here for convenience.

\begin{theorem}[Roura's Continuous Master Theorem (CMT)]
\label{thm:CMT}
	Let \(F_n\) be recursively defined~by
	\begin{align}
	\label{eq:CMT-recurrence}
		F_n \wwrel= \begin{dcases*}
			b_n\:,	
				& for \( 0 \le n < N \); \\
			\vphantom{\bigg|}
				t_n \bin+ \smash{\sum_{j=0}^{n-1} w_{n,j} \, F_j}, 
				& for  \(n \ge N\)\,, 
		\end{dcases*}
	\end{align}
	where \(t_n\), the toll function, satisfies \(t_n \sim K n^\alpha \log^\beta(n)\) as
	\(n\to\infty\) for constants \(K\ne0\), \(\alpha\ge0\) and \(\beta > -1\).
	Assume there exists a function \(w:[0,1]\to \R_{\ge0}\), the \textit{shape function,}
	with \(\int_0^1 w(z) dz \ge 1 \) and
	\begin{align}
	\label{eq:CMT-shape-function-condition}
		\sum_{j=0}^{n-1} \,\biggl|
			w_{n,j} \bin- \! \int_{j/n}^{(j+1)/n} \mkern-15mu w(z) \: dz
		\biggr|
		\wwrel= \Oh(n^{-d}),
		\qquad(n\to\infty),
	\end{align}
	for a constant \(d>0\).
	With \(\displaystyle H \ce 1 - \int_0^1 \!z^\alpha w(z) \, dz\), we
	have the following cases:
	\begin{enumerate}[itemsep=0ex]
		\item If \(H > 0\), then \(\displaystyle F_n \sim \frac{t_n}{H}\).
			\(\vphantom{\displaystyle\int_0^1}\)
		\item \label{case:CMT-H0} 
		If \(H = 0\), then 
		$\displaystyle
		F_n \sim \frac{t_n \ln n}{\tilde H}$ with 
		\(\displaystyle \tilde H = -(\beta+1)\int_0^1 \!z^\alpha \ln(z) \, w(z) \, dz\).
		\item \label{case:CMT-theta-nc}
		If \(H < 0\), then \(F_n = \Oh(n^c)\) for the unique
		\(c\in\R\) with \(\displaystyle\int_0^1 \!z^c w(z) \, dz = 1\).
	\end{enumerate}
\qed\end{theorem}

\noindent
\wref{thm:CMT} is the ``reduced form'' of the CMT,
which appears as Theorem~1.3.2 in Roura's doctoral thesis~\cite{Roura1997},
and as Theorem~18 of \cite{MartinezRoura2001}.
The full version (Theorem~3.3 in~\cite{Roura2001})
allows us to handle sublogarithmic factors in the toll function, as well, 
which we do not need here.

\section{Local Limit Law for the Beta-Binomial Distribution}
\label{app:local-limit-beta}

Since the binomial distribution is sharply concentrated, one can
use Chernoff bounds on beta-binomial variables after conditioning on the beta distributed
success probability.
That already implies that $\betaBinomial(n,\alpha,\beta)/n$ converges to $\betadist(\alpha,\beta)$
(in a specific sense).
We can obtain stronger error bounds, though, by directly comparing the PDFs.
Doing that gives the following result;
a detailed proof is given in \cite{Wild2016}, Lemma~2.38.

\begin{lemma}[Local Limit Law for Beta-Binomial, \cite{Wild2016}, Lemma~2.38]
\label{lem:beta-binomial-convergence-to-beta}
~\\
	Let \((\ui In)_{n\in\N_{\ge1}}\) 
	be a family of random variables with beta-binomial distribution,
	\(\ui In \eqdist \betaBinomial(n, \alpha,\beta)\) where \(\alpha,\beta\in\{1\}\cup\R_{\ge2}\), 
	and let \(f_B(z)\) be the density of the \(\betadist(\alpha,\beta)\) distribution.
	Then we have uniformly in \(z\in(0,1)\) that 
	\begin{align*}
				n \cdot \Prob[\big]{ I = \lfloor z(n+1)\rfloor } 
			\wwrel= 
				f_B(z) \bin\pm \Oh(n^{-1})
				,\qquad (n\to\infty).
	\end{align*}
	That is, \(\ui In/n\) converges to \(\betadist(\alpha,\beta)\) in distribution, 
	and the probability weights converge uniformly to the limiting density at rate \(\Oh(n^{-1})\).
\end{lemma}

\section{Smoothness of the Shape Function}
\label{app:shape-function-condition}

In this appendix we show
that $w(z)$ as given in \wpeqref{eq:shape-function}
fulfills \wpeqref{eq:CMT-shape-function-condition}, the approximation-rate criterion of the CMT.
We consider the following ranges for 
$\frac{\lfloor zn\rfloor}{n-1} = \frac j{n-1}$ separately:
\begin{itemize}
\item $\frac{\lfloor zn\rfloor}{n-1} < \frac13$ and $\frac12<\frac{\lfloor zn\rfloor}{n-1}<\frac23$.\\
	Here $w_{n,\lfloor zn\rfloor} = 0$ and so is $w(z)$.
	So actual value and approximation are exactly the same.
\item $\frac13<\frac{\lfloor zn\rfloor}{n-1} < \frac12$ and $\frac{\lfloor zn\rfloor}{n-1}>\frac23$.\\
	Here $w_{n,j} = 2 \Prob{J=j}$ and $w(z) = 2 f_P(z)$ where $f_P(z) = z^t(1-z)^t / \BetaFun(t+1,t+1)$
	is twice the density of the beta distribution $\betadist(t+1,t+1)$.
	Since $f_P$ is Lipschitz-continuous on the bounded interval $[0,1]$ (it is a polynomial)
	the uniform pointwise convergence from above is enough to bound the sum of
	$\bigl| w_{n,j} \bin- \! \int_{j/n}^{(j+1)/n} w(z) \: dz \bigr|$ over all $j$ in the range by
	$\Oh(n^{-1})$.
\item $\frac{\lfloor zn\rfloor}{n-1} \in \{\frac13,\frac12,\frac23\}$.\\
	At these boundary points, the difference between $w_{n,\lfloor zn\rfloor}$ and $w(z)$ does not
	vanish (in particularly $\frac12$ is a singular point for $w_{n,\lfloor zn\rfloor}$), 
	but the absolute difference is bounded.
	Since this case only concerns $3$ out of $n$ summands, the overall contribution to the error
	is $\Oh(n^{-1})$.
\end{itemize}
Together, we find that \weqref{eq:CMT-shape-function-condition} is fulfilled as claimed:
\begin{align}
	\sum_{j=0}^{n-1} \,\biggl|
		w_{n,j} \bin- \! \int_{j/n}^{(j+1)/n} \mkern-15mu w(z) \: dz
	\biggr|
	\wwrel= \Oh(n^{-1})
	\qquad(n\to\infty).
\end{align}

\section{Approximation by (Incomplete) Beta Integrals}
\label{app:approx-by-beta-integrals}

\begin{lemma}
\label{lem:E-Jn-int}
	Let $J \eqdist \betaBinomial(n-c_1,\alpha,\beta) + c_2$ be a random variable
	that differs by fixed constants $c_1$ and $c_2$ from a beta-binomial variable 
	with parameters $n\in \N$ and $\alpha, \beta\in\N_{\ge1}$.
	Then the following holds
	\begin{thmenumerate}{lem:E-Jn-int}
	\item
	\label{lem:E-xy-Jn}
		For fixed constants $0\le x\le y\le 1$ holds
		\begin{align*}
				\E[\big]{[x n\le J\le y n] \cdot J \lg J}
			&\wwrel=
				\frac{\alpha}{\alpha+\beta} \, I_{x,y}(\alpha+1,\beta) \cdot n\lg n \wbin\pm\Oh(n)
				,\qquad(n\to\infty).
		\end{align*}
		The result holds also when any or both of the inequalities in $[x n\le J\le y n]$ are strict.
	\item 
	\label{lem:E-Jn-ln-Jn}
		$\E{\frac Jn \ln \frac Jn} = 
		\frac{\alpha}{\alpha+\beta} (\harm{\alpha}-\harm{\alpha+\beta}) \pm \Oh(n^{-h})$
		for any $h\in(0,1)$.
	\end{thmenumerate}
\end{lemma}

\begin{proof}
We start with part (a).
By the local limit law for beta binomials (\wref{lem:beta-binomial-convergence-to-beta}) 
it is plausible to expect a reasonably small error when
we replace $\E[\big]{[x n\le J\le y n] \cdot J \lg J}$ by
$\E[\big]{ [x \le P\le y ] \cdot (P n) \lg (P n) }$
where $P\eqdist \betadist(\alpha,\beta)$ is beta distributed.
We bound the error in the following.

We have 
$
	\E[\big]{[x n\le J\le y n] \cdot J \ln J}
	\wrel=
	\E[\big]{[x n\le J\le y n] \cdot \frac Jn} \cdot n \ln n \pm \Oh(n)
$
by \weqref{eq:ex-JlnJ};
it thus suffices to compute $\E[\big]{[x n\le J\le y n] \cdot \tfrac Jn}$.
We first replace $J$ by $I\eqdist\betaBinomial(n,\alpha,\beta)$ and argue later that this 
results in a sufficiently small error.
We expand
\begin{align*}
		\E[\big]{[x\le \tfrac In\le y] \cdot \tfrac In}
	&\wwrel=
		\sum_{i = \lceil xn\rceil} ^ {\lfloor yn \rfloor}
			\tfrac in \cdot \Prob{I = i}
\\	&\wwrel=
		\frac 1n \sum_{i = \lceil xn\rceil} ^ {\lfloor yn \rfloor}
			\tfrac in \cdot n \Prob{I = i}
\\	&\wwrel{\relwithref[r]{lem:beta-binomial-convergence-to-beta}=}
		\frac 1n \sum_{i = \lceil xn\rceil} ^ {\lfloor yn \rfloor}
			\tfrac in \cdot 
			\biggl(\frac{(i/n)^{\alpha-1}(1-(i/n))^{\beta-1}}{\BetaFun(\alpha,\beta)} \bin\pm \Oh(n^{-1})\biggr)
\\	&\wwrel=
		\frac1 {\BetaFun(\alpha,\beta)} \cdot \frac 1n \sum_{i = \lceil xn\rceil} ^ {\lfloor yn \rfloor}
			f(i/n) \wwbin\pm \Oh(n^{-1}),
\end{align*}
where 
$f(z) = z^{\alpha} (1-z)^{\beta-1}$.

Note that $f(z)$ is Lipschitz-continuous on the bounded interval $[x,y]$ 
since it is continuously differentiable (it is a polynomial).
Integrals of Lipschitz functions are well-approximated by finite Riemann sums;
see Lemma 2.12 (b) of \cite{Wild2016} for a formal statement.
We use that on the sum above
\begin{align*}
		\frac 1n \sum_{i = \lceil xn\rceil} ^ {\lfloor yn \rfloor}
					f(i/n)
	&\wwrel=
		\int_x^y f(z) \,dz  \wbin\pm \Oh(n^{-1})
		,\qquad(n\to\infty).
\end{align*}
Inserting above and using $\BetaFun(\alpha+1,\beta) / \BetaFun(\alpha,\beta) = \alpha/(\alpha+\beta)$
yields
\begin{align*}
		\E[\big]{[x \le \tfrac In \le y] \cdot \tfrac In}
	&\wwrel=
		\frac{\int_x^y z^{\alpha}(1-z)^{\beta-1}\, dz}{\BetaFun(\alpha, \beta)}
		\wwbin\pm \Oh(n^{-1})
\\	&\wwrel=
		\frac{\alpha}{\alpha+\beta} \, I_{x,y}(\alpha+1,\beta) 
		\wwbin\pm \Oh(n^{-1});
\numberthis\label{eq:E-I-xy}
\shortintertext{recall that}
		I_{x,y}(\alpha,\beta)
	&\wwrel=
		\int_x^y \frac{z^{\alpha-1}(1-z)^{\beta-1}}{\BetaFun(\alpha,\beta)} \, dz
	\wwrel=
		\Prob[\big]{ x<P<y }
\end{align*}
denotes the regularized incomplete beta function.

Changing from $I$ back to $J$ has no influence on the given approximation.
To compensate for the difference in the number of trials ($n-c_1$ instead of $n$),
we use the above formulas for with $n-c_1$ instead of $n$; since we let $n$ go to
infinity anyway, this does not change the result.
Moreover, replacing $I$ by $I+c_2$ changes the value of the argument $z=I/n$
of $f$ by $\Oh(n^{-1})$; 
since $f$ is smooth, namely Lipschitz-continuous,
this also changes $f(z)$ by at most $\Oh(n^{-1})$.
The result is thus not affected by more than the given error term:
\[
		\E[\big]{[x\le \tfrac Jn\le y] \cdot \tfrac Jn}
	=
		\E[\big]{[x\le \tfrac In\le y] \cdot \tfrac In}
		\pm \Oh(n^{-1})
\]
We obtain the claim by multiplying with $n\lg n$.

Versions with strict inequalities in $[x n\le J\le y n]$
only affect the bounds of the sums above by one, which again gives a negligible error of 
$\Oh(n^{-1})$.

\medskip\noindent
This concludes the proof of part~(a).

\needspace{5\baselineskip}
\bigskip\noindent
For part (b), we follow a similar route. 
The function we integrate is no longer Lipschitz continuous,
but a weaker form of smoothness is sufficient to bound the difference between
the integral and its Riemann sums.
Indeed, the above cited Lemma 2.12 (b) of \cite{Wild2016} is formulated 
for the weaker notion of Hölder-continuity:
	A function $f:I\to R$ defined on a bounded interval $I$ is called Hölder-continuous with exponent $h\in(0,1]$
	when 
	\[
		\exists C\;
		\forall x,y\in I\wrel:
			\bigl| f(x) - f(y) \bigr|
			\wrel\le 
			C |x-y|^h.
	\]
	This generalizes Lipschitz-continuity (which corresponds to $h=1$).
	
As above, we replace $J$ by $I\eqdist\betaBinomial(n,\alpha,\beta)$,
which affects the overall result by $\Oh(n^{-1})$.
We compute
\begin{align*}
		\E[\big]{\tfrac In \ln \tfrac In}
	&\wwrel=
		\sum_{i = 0} ^ {n}
			\tfrac in \ln \tfrac in \cdot \Prob{I = i}
\\	&\wwrel{\relwithref[r]{lem:beta-binomial-convergence-to-beta}=}
		\frac 1n \sum_{i = 0} ^ {n}
			\tfrac in \ln \tfrac in \cdot 
			\biggl(\frac{(i/n)^{\alpha-1}(1-(i/n))^{\beta-1}}{\BetaFun(\alpha,\beta)} \bin\pm \Oh(n^{-1})\biggr)
\\	&\wwrel=
		- \frac1 {\BetaFun(\alpha,\beta)} \cdot \frac 1n \sum_{i = 0} ^ {n}
			f(i/n) \wwbin\pm \Oh(n^{-1}),
\end{align*}
where now
$f(z) = \ln (1/z) \cdot z^{\alpha} (1-z)^{\beta-1}$.
Since the derivative is $\infty$ for $z=0$, $f$ cannot be Lipschitz-continuous,
but it is Hölder-continuous
on $[0,1]$ for any exponent $h \in (0,1)$:
$z\mapsto \ln(1/z)z$ is Hölder-continuous 
(see, \eg, \cite{Wild2016}, Prop.~2.13.), 
products of Hölder-continuous function remain such on bounded intervals
and the remaining factor of $f$ is a polynomial in $z$, which is Lipschitz- and
hence Hölder-continuous.

By Lemma 2.12 (b) of \cite{Wild2016} we then have
\begin{align*}
		\frac 1n \sum_{i = 0} ^ {n}
					f(i/n)
	&\wwrel=
		\int_0^1 f(z) \,dz  \wbin\pm \Oh(n^{-h})
\end{align*}
Recall that we can choose $h$ as close to $1$ as we wish; this will
only affect the constant hidden by the $\Oh(n^{-h})$.
It remains to actually compute the integral;
fortunately, this ``logarithmic beta integral'' has a well-known closed form
(see, \eg,~\cite{Wild2016}, Eq.~(2.30)).
\begin{align*}
		\int_0^1 \ln (z) \cdot z^{\alpha} (1-z)^{\beta-1}
	&\wwrel=
		\BetaFun(\alpha+1,\beta) \bigl( \harm{\alpha} - \harm{\alpha+\beta} \bigr)
\end{align*}
Inserting above, we finally find
\begin{align*}
		\E{\tfrac Jn \ln \tfrac Jn}
	&\wwrel=
		\E{\tfrac In \ln \tfrac In} \wbin\pm\Oh(n^{-1})
\\	&\wwrel=
		\frac{\alpha}{\alpha+\beta} \bigl( \harm{\alpha} - \harm{\alpha+\beta} \bigr)
			\wbin\pm\Oh(n^{-h})
\end{align*}
for any $h \in (0,1)$.
\end{proof}

\needspace{8\baselineskip}

\bibliographystyle{plainurl}
\bibliography{quick-mergesort}

\end{document}